\documentclass[a4paper,UKenglish,cleveref, autoref, thm-restate]{lipics-v2021}

\usepackage{todonotes}
\usepackage{xspace}
\usepackage{appendix}
\usepackage{enumitem}

\def\NN{\mathbb{N}}	

\def\CP{\mathcal{P} }
\def\CQ{\mathcal{Q} }

\newcommand{\cost}{\textrm{cost}}

\newcommand{\STC}{\mathrm{STC}}     
\newcommand{\CMC}{\mathrm{MCMC}}

\newcommand{\initOneLiners}{%
    \setlength{\itemsep}{0pt}
    \setlength{\parsep }{0pt}
    \setlength{\topsep }{0pt}
}

\usepackage{pgffor}
\foreach \x in {A,...,Z}{%
	\expandafter\xdef\csname b\x\endcsname{\noexpand\ensuremath{\noexpand\mathbf{\x}}}
	\expandafter\xdef\csname c\x\endcsname{\noexpand\ensuremath{\noexpand\mathcal{\x}}}
	\expandafter\xdef\csname B\x\endcsname{\noexpand\ensuremath{\noexpand\mathbb{\x}}}
}

\foreach \x in {a,...,z}{%
	\expandafter\xdef\csname b\x\endcsname{\noexpand\ensuremath{\noexpand\mathbf{\x}}}
}

\newcommand{\Oh}[1]{\ensuremath{\cO}\left(#1\right)}
\def\CC{\ensuremath{\mathcal{C}}}

\newtheorem*{problem*}{Problem}

\title{Min-Max Connected Multiway Cut} 

\author{Petr Kolman}{Charles University, Faculty of Mathematics and Physics, Department of Applied Mathematics, Czechia \and \url{https://kam.mff.cuni.cz/~kolman/} }{kolman@kam.mff.cuni.cz}{https://orcid.org/0000-0003-2235-0506}{Partially supported by grant 24-10306S of GA \v{C}R.}

\author{Hans Raj {Tiwary}}{Charles University, Faculty of Mathematics and Physics, Department of Applied Mathematics, Czechia \and \url{https://kam.mff.cuni.cz/~hansraj/}}{hansraj@kam.mff.cuni.cz}{https://orcid.org/0000-0003-1903-1600}{Partially supported by the AGATE project funded  from the Horizon Europe Programme under Grant Agreement No. 101183743.}

\authorrunning{P. Kolman and H.\,R. Tiwary} 

\Copyright{Petr Kolman and Hans Raj Tiwary} 



\category{} 

\relatedversion{} 

\acknowledgements{}

\nolinenumbers 

\begin{document}
\setuptodonotes{color=red!20}

\maketitle

\begin{abstract}
We introduce a variant of the multiway cut that we call the min-max connected multiway cut. Given a graph $G=(V,E)$ and a set $\Gamma\subseteq V$ of $t$ terminals, partition $V$ into $t$ parts such that each part is connected and contains exactly one terminal; the objective is to minimize the maximum weight of  the edges leaving any
part of the partition. This problem is a natural modification of the standard multiway cut problem and
it differs from it in two ways: first, the cost of a partition is defined to be the maximum size of the boundary of any part, as opposed to the sum of all boundaries, and second, the subgraph induced by each part is required to be connected. Although the modified objective function has been considered before in the literature under the name min-max multiway cut, the requirement on each component to be connected has not been studied as far as we know. Our main motivation for studying this problems is its close connection with the spanning tree congestion problem that has been extensively studied but on which little progress has been made.

We show various hardness results for this problem, including a proof of weak NP-hardness of
the weighted version of the problem on graphs with treewidth two, and $W[1]$-hardness for the problem when parameterized by the treewidth of the graph. Complementing our lower bounds tightly, we also provide a pseudopolynomial time algorithm as well as an FPTAS for the weighted problem on graphs of bounded treewidth. As a consequence of our investigation we also show that the (unconstrained) min-max multiway cut problem is NP-hard even for three terminals, strengthening the known results.
\end{abstract}

\section{Introduction}\label{sec:introduction}
Cuts in graphs are a fundamental concept in computer science, playing a crucial role in both theory and numerous applications~\cite{Shmoys:97}. They are used for defining important graph parameters such as bisection, expansion, and treewidth, and have stimulated the development of various algorithmic techniques, including linear programming, divide and conquer, and semidefinite relaxations. Graph cuts are employed in diverse areas such as computer graphics, parallel and distributed systems, 
transportation, military and agriculture applications, analysis of social networks, and VLSI design, to name a few.

Generally speaking, in usual cut problems, the goal is to \emph{separate} some (pairs of) vertices, and it is irrelevant what happens to the other pairs - whether they end up in the same or in different connected components. However, in applications such as image segmentation, route planning, power grid management, forest planning and harvest scheduling, and market zoning, it is
important that parts of $G$ obtained by the cut are \emph{connected}~\cite{DGPW:11,HRP:11,CCG:13,VKR:08,GKLSZ:19}.


The Min-Max Connected Multiway Cut problem combines these two orthogonal requirements: 
separating certain graph parts while keeping others connected.
This kind of combination of separation and connectivity constraints
is a natural extension of classical problems and brings interesting
challenges. 
Though each of the two constraints is relatively well-understood 
\emph{individually}, and in some cases, the fulfillment of the
connectivity requirement comes for free as a bonus of the optimality of the cut
objective (see Related Results), the \emph{combination} of the requirements in the general setting
makes the problems challenging. Little theoretical research has addressed both 
requirements simultaneously. 

Our main motivation for studying the min-max connected multiway cut problem
is its close connection to the Spanning Tree Congestion problem (See Sec.~\ref{subsec:stc}), a problem
that has been studied extensively but with very little algorithmic progress \cite{kolman:25}.

\subsection{Terminology}
A family of disjoint subsets 
$S_1,\ldots,S_k\subseteq V$ is a {\em partition} of the set $V$ if $\bigcup_{i=1}^k S_i =V$.
For a graph $G=(V,E)$ and a subset $S\subseteq V$
of vertices, $G[S]$ is the subgraph of $G$ induced by the subset $S$;
by $n=|V|$ we denote the number of vertices and by $m=|E|$ the number of edges of $G$. 
For a partition $S_1,\ldots,S_k$ of the vertex set of a graph $G=(V,E)$,
$E(S_1,\ldots,S_k)=\left \{ \{u,v\}\ | \ u\in S_i, v\in S_j, i\not =j \right \}$
denotes the set of edges between the parts \(S_1,\ldots,S_k\).
For a positive integer $k$, $[k]=\{1,\ldots,k\}$.

Given a graph $G=(V,E)$ and a subset $\Gamma=\{t_1,\ldots,t_k\} $ of $k$ 
vertices, called {\em terminals}, a {\em Connected Multiway Cut} is a partition 
$S_1,\ldots, S_k$ of $V$ such that
\begin{enumerate}
\item for each $i\in [k]$, $t_i\in S_i$, and 
\item for each $i\in [k]$, $G[S_i]$ is connected.
\end{enumerate}

Without the second requirement, we simply have a Multiway Cut~\cite{BFK:14,SvitkinaT:04}. 
The \emph{cost of a Multiway Cut} $C=\{S_1,\ldots,S_k\}$ of a graph $G=(V,E)$ 
with edge weights $w:E\rightarrow \NN^+$ 
is $\cost(C)=\max_{i}\delta(S_i)$ where $\delta(S)=\sum_{e\in E(S,V\setminus S)}w(e)$ denotes
the sum of weights of edges between $S$ and $V\setminus S$. 
We define the \emph{cost of a Connected Multiway Cut} analogously.
The {\em Min-Max Connected Multiway Cut} problem is to find a Connected Multiway Cut
of minimum cost; the minimum cost is denoted $\CMC(G,T)$.
In the unweighted version of the problem, all edges have weight one.

\subsection{Our Results}
The NP-hardness of the min-max connected multiway cut on general graphs with at least four terminals follows using the proof of the NP-hardness of the min-max multiway cut by Svitkina and Tardos \cite{SvitkinaT:04} without any modification. Therefore, we focus our attention on either simpler graphs or fewer number of terminals. In particular, we show the following new results.
\begin{itemize}
    \item We show that the min-max connected multiway cut problem is weakly NP-hard for three terminals on graphs of treewidth three. We also show the same weak NP-hardness for the (unconstrained) min-max multiway cut problem.
    \item We show that the min-max connected multiway cut problem remains weakly NP-hard on graphs of treewidth two if one allows arbitrary number of terminals.
    \item We show that the min-max connected multiway cut problem is $W[1]$-hard when paremeterized by the treewidth, or number of terminals, or both, for the case when the input weights are given in unary.
    \item We give a pseudopolynomial algorithm for the weighted problem on graphs of bounded treewidth. Our algorithm takes time $2^{\Oh{\tau^2}}N^{\Oh{\tau}}$ where the input size (in unary) is $N$ and the graph has treewidth $\tau$. This complements our hardness results, both (weak) NP-hardness and $W[1]$-hardness.
    \item We give an FPTAS for the weighted problem on bounded treewidth graphs.
    \item We show that if the number of terminals is fixed, then the weighted problem is polynomial time solvable on trees. We do this by showing that the problem is FPT on trees when parameterized by the number of terminals.
    \item We give some evidence why the weighted problem might be hard even on trees by showing superpolynomial extension complexity for a natural polytope associated with the problem, thereby ruling out many natural LP formulations of polynomial size, and further, we show that an exact version of the weighted problem is NP-hard on trees.
\end{itemize}

\subsection{Related Results}

One of the most studied cut problems is the minimum $s-t$ cut problem. 
Given a graph $G=(V,E)$
and two distinct vertices $s$ and $t$, find a partition $S,V\setminus S$ of the vertex 
set $V$ such that $s\in S$, $t\in V\setminus S$, and $E(S,V\setminus S)$ is minimized. Given a minimum cut, we observe that both parts $S$ and $V\setminus S$ are connected; otherwise, 
there would be a strictly smaller cut. 
Garg~\cite{Garg:94} provides a polyhedral
description of all $s-t$ cuts in which the $s$-side is connected, all $s-t$ cuts 
in which the $t$-side is connected, and all $s-t$ cuts in which both sides are connected;
we note that the polyhedra are unbounded.

For the global minimum cut problem where the task is to find a non-trivial partition
$S,V\setminus S$ minimizing $E(S,V\setminus S)$, the optimum solution has again both
sides connected.

The situation is different for the maximum cut problem where the optimality of 
a solution does not imply connectivity of either side. This motivated the definitions
of two related problems. In the bond problem, the task is to find a partition
$S,V\setminus S$ maximizing $E(S,V\setminus S)$ with both sides connected~\cite{DEHKKLPSS:21,KT:26}; in the maximum connected cut problem, the objective
is the same but the connectivity constraint is imposed on one side only~\cite{Hajiaghayietal:20,Schieberv:23}.
It should be noted that the connectivity constraints make the maximum cut problem
harder: in contrast to the maximum cut without the connectivity constraints, 
both versions with the connectivity constraints are NP-hard even on planar 
bipartite graphs, and the bond
problem does not admit a constant factor approximation algorithm, unless P = NP \cite{DEHKKLPSS:21}.
Both problems are FPT with respect to the treewidth \cite{DEHKKLPSS:21}.

When cutting a graph into more than two parts, the most common objective minimizes 
the total weight of deleted edges. In contrast, the min-max multiway cut problem,
mentioned earlier, minimizes the maximum weight of deleted edges adjacent to any single part.
The best known approximation algorithm for this problem, due to Bansal et al.~\cite{BFK:14},
achieves approximation ratio $O(\sqrt{\log n \log|\Gamma|)}$,
improving over $O(\log^2 n)$-approximation by Svitkina and Tardos~\cite{SvitkinaT:04}.
Kim at al.~\cite{KimPST:17} proved that the min-max multiway cut problem is FPT when parametrized by both the number of
terminals and the maximum weight of deleted edges adjacent to any single part.

\subsection{Spanning Tree Congestion Problem}\label{subsec:stc}
One of our motivations for dealing with the min-max connected multiway cut is its close 
connection to the \emph{Spanning Tree Congestion} problem which is defined as follows~\cite{Ostrovskii:04}: 
given a graph $G=(V,E)$, construct a spanning tree $T$ of $G$ minimizing its maximum edge 
congestion where the congestion of an edge $e\in T$ is the number of edges $uv\in E$ in $G$
such that the unique path between $u$ and $v$ in $T$ passes through $e$; 
the optimal value for a given graph $G$ is denoted $\STC(G)$.

The problem was introduced under different names in the late 1990s but till today, the 
computational complexity of the problem is not much understood: the problem is known to be
NP-hard even for graphs of maximum degree at most three~\cite{atalig-etal:26}, 
there is $\Oh{\Delta\cdot\log^{3/2}n}$-approximation algorithm for graphs of maximum
degree $\Delta$~\cite{kolman:25} but the best known approximation ratio for general graph is  $n/2$ only~\cite{Otachi:20}; interestingly,
this is the approximation ratio achieved by \emph{any} spanning tree. 

The following two lemmas capture the connection between the two problems.
\begin{lemma}
Let $G=(V,E)$ be a connected graph, $T$ a spanning tree minimizing the spanning tree 
congestion, $r\in V$ any non-leaf vertex of $T$, and $\Gamma=\{t_1,\ldots,t_k\}\subset V$ the set of neighbors of $r$ in $T$. Then
\[
\STC(G) \geqslant\CMC(G\setminus\{r\},\Gamma) + 1\ .
\]
\end{lemma}
\begin{proof}
Let $V_1,\ldots,V_k$ be the connected components of $T \setminus \{r\}$; note that
they constitute a connected multiway cut for the instance $(G\setminus\{r\},\Gamma)$. For each $i\in [k]$,
the spanning tree congestion of the edge $\{r,t_i\}$ in $T$ is exactly $\delta(V_i)+1$.
Thus, we obtain the desired bound
$\STC(G) \geqslant \displaystyle\max_{i=1,\ldots,k} \delta(V_i) + 1 \geqslant \CMC(G\setminus\{r\},\Gamma) + 1 \ .
$
\end{proof}

\begin{lemma}
Let $G=(V,E)$ be a connected graph, $r\in V$ a vertex of $G$, $\Gamma=\{t_1,\ldots,t_l\}\subset V$ the set of neighbors of $r$ in $G$, and $V_1, \ldots, V_l$ an optimal solution of the Min-Max 
Connected Multiway Cut instance $(G,\Gamma)$. Then
\[
\displaystyle\STC(G) \leq \CMC(G\setminus\{r\},\Gamma) + \max\left \{1,\max_{i=1,\ldots,l} \STC(G[V_i])\right \} \ .
\]
\end{lemma}
\begin{proof}
For each $i\in [l]$, let $T_i$ be an optimal spanning tree of subgraph $G[V_i]$; 
since each of these
subgraphs is connected, all these trees exist. Let $T$ be the subgraph of $G$
obtained as the union of the disjoint trees $T_1,\ldots, T_l$ with all the edges
adjacent to $r$. Then, by construction, $T$ is a spanning tree of $G$. For each
$i\in [l]$, the congestion of the edge $\{r,t_i\}$ in $T$ is $1+\CMC(G\setminus\{r\},\Gamma)$,
and the congestion of any edge in $T_i$ is at most $\STC(G[V_i]) + \CMC(G\setminus\{r\},\Gamma)$ which completes the proof.
\end{proof}

\section{Hardness Results}\label{sec:hardness}
\subsection{Min-Max Cut vs. Min-Max Connected Cut}
For at most three terminals the min-max multiway cut and the min-max connected multiway cut
problems are equivalent in the following sense.

\begin{lemma}\label{lem:connected-vs-nonconnected}
Let $G$ be a connected graph and $\Gamma$ a set of at most three terminals. For any multiway 
cut $C$ for the instance $(G,\Gamma)$, there exists a connected multiway cut $C'$ for the
same instance such that $\cost(C')\leqslant\cost(C).$
\end{lemma}
\begin{proof}
We prove the claim for three terminals. For two terminals, the proof is analogous.

Let $C$ be a multiway cut and let $V_1, V_2, V_3$ be the corresponding
partition of $V(G)$. If $G[V_i]$ is connected for each $i\in\{1,2,3\}$ the lemma holds with $C'=C$.

So, without loss of generality, we assume that $G[V_1]$ is not connected and let $W_1$ and
$W_2$ be a partition of $V_1$ such that $t_1\in W_1$ and $E(W_1,W_2)=\emptyset$. 
We are going to define a new connected partition $V'_1, V'_2, V'_3$ of $V(G)$.
Let $V'_1=W_1$. If $e(W_2,V_2)\geqslant e(W_2,V_3)$, we define $V'_2=V_2\cup W_2$ and $V'_3=W_3$,
and if $e(W_2,V_2) < e(W_2,V_3)$, we define $V'_2=V_2$ and $V'_3=V_3\cup W_2$. In both cases,
we obtain a feasible solution of the min-max multiway cut that has fewer connected components and whose cost is not higher than the cost of $(V_1,V_2,V_3)$. Applying this procedure repeatedly, we end up with a partition where each induced subgraph is connected and whose cost is at most $\cost(C)$.
\end{proof}

The above lemma implies that an optimal solution for the min-max multiway cut problem can always be assumed to be a connected multiway cut when we have at most three terminals. For two terminals, the problem reduces to the standard $s-t$ cut and therefore can be solved in polynomial time \cite{AhujaMO:93}.

For $k=4$, there is a difference between the optimal solutions of the min-max 
multiway cut and the min-max connected multiway cut. Consider the example graph on five vertices with edge weights one and two as shown in Figure~\ref{fig:minmax-four-term}. In any connected multiway cut, the vertex $u$ cannot belong to the component of $t_4$, and if it belongs to a component of any other terminal, the boundary size of this component is seven. However, the cost of the multiway cut $\{t_1\},\{t_2\},\{t_3\},\{u,t_4\}$ is only six.
\begin{figure}[htbp]
  \centering
  \includegraphics[width=0.4\textwidth]{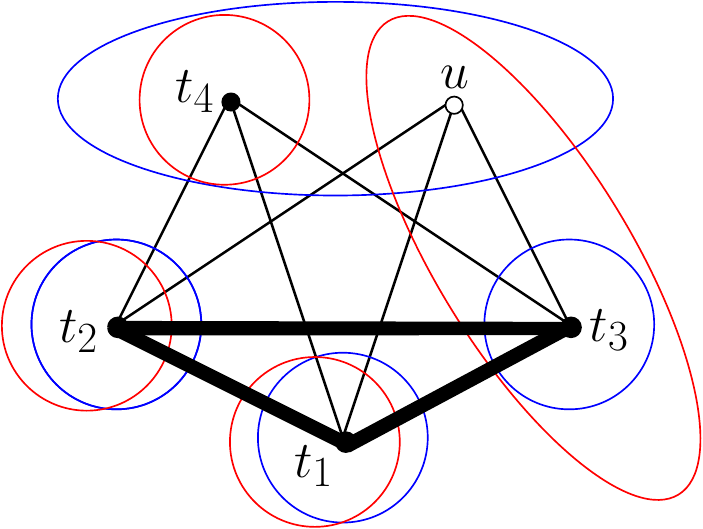}
  \caption{Min-max multiway cut vs. Min-max connected multiway cut. The fat edges are of weight two, the regular edges are of weight one. The optimal min-max cut is of cost six (the blue ovals) while the optimal min-max
  connected cut is of size seven (the red ovals).}
  \label{fig:minmax-four-term}
\end{figure}

\subsection{NP-hardness}
The weighted multiway cut problem was shown to be NP-hard for four terminals by Svitkina and Tardos \cite{SvitkinaT:04}. The same proof without any modification also proves NP-hardness of the weighted min-max connected multiway cut problem because in their reduction, Svitkina and Tardos use (four) terminals that are connected to every non-terminal. So, any multiway cut in the resulting graph is also a connected multiway cut.

For two terminals, the min-max multiway cut problem reduces to the standard $s-t$ cut problem and can be solved in polynomial time. Furthermore, from a minimum cut, a minimum connected cut can be obtained as shown in Lemma \ref{lem:connected-vs-nonconnected}.

This leaves the situation for three terminals unclear. We now prove that the weighted min-max connected multiway cut problem is weakly NP-hard for three terminals. Together with Lemma \ref{lem:connected-vs-nonconnected} this also shows weak NP-hardness of min-max multiway cut problem with three terminals. 
A problem is weakly NP-hard if it is NP-hard when the input weights are written in binary but polynomial time solvable when the input is written in unary. 


\begin{theorem}\label{thm:NPhard-general}
The weighted min-max connected multiway cut problem and the weighted min-max multiway cut problem with three terminals are weakly NP-hard.
\end{theorem}
\begin{proof}
    We show NP-hardness for the weighted min-max connected multiway problem. The NP-hardness of the weighted min-max multiway problem follows by Lemma \ref{lem:connected-vs-nonconnected}.

We give a reduction from the three-way partitioning problem, which is a special case of the multiway number partitioning problem~\cite{Korf:09} (cf. multiprocessor scheduling problem~\cite{GareyJ:79}).
In the multiway partitioning problem, we wish to separate $n$ positive integers with total sum $kN$, into $k$ parts in such a way that each part has sum $N$. This problem is weakly NP-hard for any fixed $k$.
We use $k=3$ in our reduction. 

Given $n$ positive integers $a_1,\ldots,a_n$ summing to $3N$, 
we construct a graph $G$ on $n+3$ vertices as follows: 
for each $i \in [n]$, there is a vertex $v_i$, and in addition, there are three terminals $t_1, t_2, t_3$. Each of the $n$ vertices $v_i$ is connected to each of 
the terminals by an edge of weight $a_i/2$. In any multiway cut, let 
$S_j\subseteq \{v_1,\ldots,v_n\}$ be the set of non-terminals in the component of $t_j$.
Then, the boundary $\delta(S_j)$ has total weight $\sum_{v_i\in S_j}2\cdot a_i/2 + \sum_{v_i\notin S_j}a_i/2=3N/2+\sum_{v_i\in S_j}a_i/2$. 
Thus, as $\sum_{j=1}^3\sum_{v_i\in S_j} a_i=3N$, the optimum min-max connected multiway cut 
in $G$ with terminals $t_1, t_2, t_3$ has size at least $2N$, and it equals $2N$ if and only if there is a three-way partition of the input numbers with each part having sum exactly~$N$.   
\end{proof}
The graphs used in the proof of Theorem \ref{thm:NPhard-general} are the complete
bipartite graphs $K_{3,n}$, and these graphs have treewidth three. This fact, together with Lemma \ref{lem:connected-vs-nonconnected}, gives us the 
following corollary.

\begin{corollary}
The weighted min-max connected multiway cut and the weighted min-max  
multiway cut with three terminals are weakly NP-hard on graphs of treewidth at least three.
\end{corollary}

If we allow more terminals, then the NP-hardness holds for simpler graphs, namely for
graphs of treewidth two. 

\begin{theorem}\label{thm:np-hard-bipartite}
The weighted min-max connected multiway cut on bipartite planar graphs of treewidth two is weakly NP-hard.
\end{theorem}
\begin{proof}
    We use a straightforward modification of the NP-hardness proof for the min-max multiway cut on trees by Svitkina and Tardos \cite{SvitkinaT:04}. We give a reduction from the Partition problem which is weakly NP-hard \cite{Karp:72}.

    Given $n$ positive integers $a_1,\ldots,a_n$ with total sum $2B$, we want to check whether they can be partitioned into two subsets, each summing to $B$. 
    For each $i\in[n]$, we create a tree $\tau_i$ with a non-terminal vertex $v_i$ and six terminals, with edge weights as shown in Figure~\ref{fig:np-hardness-k=2}. Finally, we add two terminals $t_1, t_2$, and connect them with each non-terminal $v_i$ with an edge of weight zero.

   \begin{figure}[htbp]
      \centering
      \includegraphics[width=0.35\textwidth]{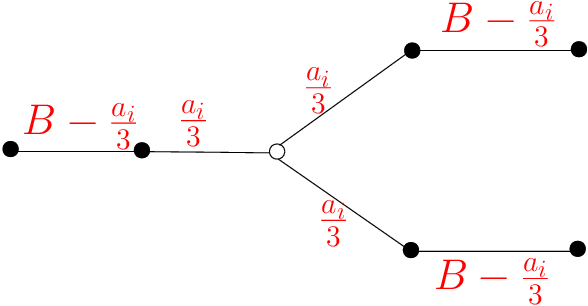}
      \caption{Gadget for NP-hardness for treewidth two}
      \label{fig:np-hardness-k=2}
    \end{figure}

Note that $K_{2,n}$ has treewidth two, and attaching disjoint trees to
it does not increase it. Thus, the resulting graph has treewidth two.
    By construction, it is clearly 
    a bipartite planar graph.

    In any connected multiway cut $C$, if a non-terminal $v_i$ is assigned to any terminal other than $t_1,t_2$, then $\cost(C) > B$. If every $v_i$ is assigned to $t_1$ or $t_2$, then 
    \[
        \cost(C)=\max\left\{\sum_{v_i \text{ assigned to } t_1}a_i, \sum_{v_i \text{ assigned to } t_2}a_i\right\} \ .\] 
    Thus, an optimal min-max connected multiway cut has cost $B$ if and only if there is a desired partition of the given numbers.
\end{proof}

\subsection{Parametrized Complexity}
 In the previous subsection we saw that the min-max connected multiway cut problem is weakly NP-hard
 on graphs of bounded treewidth. Complementing this, in Section \ref{sec:algorithms} we will show that the problem can be solved in pseudopolynomial time, that is, in time $2^{\Oh{\tau^2}}N^{\Oh{\tau}}$ where $\tau$ is the treewidth of the graph and $N$ is the input size assuming the weights are given in unary.

Now we show that this dependence on the treewidth in the exponent is perhaps unavoidable by showing that the min-max connected multiway cut problem is $W[1]$-hard when parameterized by the treewidth of the graph.

\begin{theorem}\label{thm:w1-hardness}
    The min-max connected multiway cut problem is $W[1]$-hard when parameterized by the treewidth of the input graph even if the input weights are given in unary.
\end{theorem}
\begin{proof}
    We generalize the reduction used in the proof of Theorem \ref{thm:NPhard-general} to number of parts $k$. The $k$-way number partitioning problem is known to be $W[1]$-hard when parameterized by $k$ under the name Exact Unary Bin Packing problem \cite{BlazejJRS25, JansenKMS13}.

    Given positive integers $a_1,\ldots,a_n$ summing to $kN$ we wish to find whether there is a partition $S_1,\ldots,S_k$ of $[n]$ such that $\sum_{j\in S_i}=N$ for each $i\in [k]$. We take the complete bipartite graph $K_{n,k}$ with $n$ non-terminals $\{v_1,\ldots,v_n\}$ and $k$ terminals $\{t_1,\ldots,t_k\}$. For a non-terminal $v_i$ and a terminal $t_j$ we assign the weight $a_i/(k-1)$.

    In any connected multiway cut, let $S_j$ be the set of non-terminals assigned to terminal $t_j$. Then, the boundary $\delta(S_j)$ has total weight $\sum_{v_i\in S_j}(k-1)\cdot a_i/(k-1)+\sum_{v_i\notin S_j}a_i/(k-1)=kN/(k-1)+(k-2)/(k-1)\sum_{v_i\in S_j}a_i.$ Thus the optimum min-max connected multiway cut in our graph has size at least $2N$ and it is equal to $2N$ if and only if there is a $k$-way partition of the input numbers with each part summing exactly to $N$.

    The treewidth of our graph is $k$ assuming $k\leqslant n.$ So min-max connected multiway cut is $W[1]$-hard when parameterized by the treewidth.
\end{proof}

In the above construction one can replace an edge with (positive) integer weight $w$ by $w$ parallel paths each of length two and have an unweighted graph of treewidth $k$. 

\begin{theorem}
    The unweighted min-max connected multiway cut problem is $W[1]$-hard when parameterized by the treewidth of the graph.    
\end{theorem}

Further note, that in our construction, the number of terminals is $k$, so we get similar
hardness result when parameterized by the number of terminals, or by both the number of terminals and the treewidth together. Also note, that the pathwidth of the graph in the hardness proof is $k$; thus, the $W[1]$-hardess results apply to the parameter pathwidth of the graph as well.

\subsection{Extension Complexity}
Even though we do not know whether the weighted min-max connected multiway cut problem is NP-hard on trees, now we give some hardness results that are closely related. 
More precisely, we now show that a natural polytope associated with the problem has superpolynomial extension complexity.

Let $P\subset \mathbb{R}^d$ be a polytope. A polytope $Q\subset\mathbb{R}^{d+r}$ is called an extended formulation if $P$ is a projection of $Q$. That is, $$P=\left\{x\in\mathbb{R}^d~\left|~\exists y\in\mathbb{R}^d, \binom{x}{y}\in Q\right.\right\}.$$

The extension complexity of a polytope $P$ -- denoted by $\mathrm{xc}(P)$ -- is the smallest number of inequalities needed to describe any extended formulation of $P.$ A superpolynomial lower bound on the extension complexity of a polytope implies that one cannot obtain a polynomial size linear program to optimize a linear function over $P$ using auxiliary variables \cite{ConfortiCZ13, FioriniMPTW15}. This rules out a wide range of algorithms even though it does not completely rule out polynomial algorithms \cite{Rothvoss17}.

Let $G=(V,E)$ be a graph, $w:E\to\mathbb{R}$ be a weight function on the edges, and $\Gamma\subseteq V$
be a set of terminals. We denote by $\mathcal{C}$ the set of all connected multiway cuts in $G$ for $T$, and for a connected cut $C\in \mathcal{C}$ we denote its cost, that is, the maximum over the sizes of the boundaries of its connected components, by $\nu(C),$ and by $\chi^C\in\{0,1\}^{|E|}$ we denote the characteristic vector of $C$ . We define the Min-Max Connected cut polytope $P(G,\Gamma,w)$ to be the convex hull of $\{(\chi^C,\nu(C)\in\mathbb{R}^{|E|+1}~|~C\in \CC\}.$ In other words, $P(G,\Gamma,w)$ is the convex hull of the characteristic vectors of all connected cuts in $G$ extended with a coordinate specifying its value\footnote{This extra-coordinate as well as definition of $P$ to include the weight function, seems unavoidable as the min-max objective function is not linear.}. 

We will prove a superpolynomial lower bound on the extension complexity of $P(G,\Gamma,w)$. To do this, we will use a superpolynomial lower bound on the extension complexity of the convex hull of feasible solutions of the partition problem.

Let $S=\{a_1,\ldots,a_n\}$ be a set of positive integers, and let $B\in\mathbb{N}$ such that $\sum_{i\in[n]}a_i=2B.$ The partition polytope $\mathrm{PARTITION(S,B)}$ is the convex hull of the characteristic vectors of subsets $I\subset [n]$ such that $\sum_{i\in I}a_i = B.$

\begin{theorem}\label{thm:xc_lb_partition}
For infinitely many integers $n\in\mathbb{N}$ there exists a set $S=\{a_1,\ldots,a_n\}$ of positive integers, $B\in\mathbb{N}$ with $\sum_{i\in[n]}a_i=2B$ such that the extension complexity of $\mathrm{PARTITION(S,B)}$ is superpolynomial in $n$.
\end{theorem}
\begin{proof}
        Let $\hat{S}=\{a_1,\ldots,a_n\}$ be a set of positive integers, and let $\hat{B}\in\mathbb{N}$. The SUBSET-SUM problem asks whether there exists a subset $I\subseteq [n]$ such that $\sum_{i\in I}a_i=\hat{B}$. We will call the convex hull of the characteristic vectors of the feasible solutions of the instance $(\hat{S},\hat{B})$ the polytope $\mathrm{SUBSETSUM(\hat{S},\hat{B}).}$ It is known \cite{AvisTiwary:15} that there are infinitely many values of $n\in\mathbb{N}$ so that there exists an instance $(\hat{S},\hat{B})$ of SUBSET-SUM problem with $|\hat{S}|=n$ such that $\mathrm{xc}(\textrm{SUBSETSUM}(\hat{S},\hat{B}))$ is superpolynomial in $n.$

    Given an instance $(\hat{S},\hat{B})$ of the SUBSET-SUM problem we use a standard reduction \cite{GareyJ:79} to PARTITION problem to create an instance $(S,B)$ of the partition problem such that $\mathrm{xc}(\mathrm{SUBSETSUM}(\hat{S},\hat{B}))\leqslant \mathrm{xc}(\mathrm{PARTITION}(S,B)),$ thus proving the theorem.

    Given $(\hat{S},\hat{B})$ define $S=\hat{S}\cup\{a_{n+1},a_{n+2}\}$ where $a_{n+1}=\sum_{i=1}^na_i$ and  $a_{n+2}=2\hat{B}$, and $B=\hat{B}+\sum_{i=1}^na_i$. The partition instance $(S,B)$ has a solution if and only if the subset sum instance $(\hat{S},\hat{B})$ has a solution. If $S'\subseteq [n]$ is a feasible solution of the subset sum instance, then $S'\cup\{n+1\}$ is a feasible solution of the partition instance. Furthermore, if $S'\subseteq[n+2]$ is a feasible solution of the partition instance such that $n+1\in S'$ then $S'\setminus\{n+1\}$ is a feasible solution of the subset sum instance. Therefore, $\mathrm{SUBSETSUM}(\hat{S},\hat{B})$ is a face of $\mathrm{PARTITION}(S,B)$ obtained using the valid inequality $x_{n+1}\leqslant 1.$ Therefore, $\mathrm{xc}(\mathrm{SUBSETSUM}(\hat{S},\hat{B}))\leqslant \mathrm{xc}(\mathrm{PARTITION}(S,B)).$
\end{proof}

\begin{theorem}\label{thm:xc_lb}
    For infinitely many integers $n\in\mathbb{N}$ there exists a tree $T$ on $n$ vertices, a set of terminals $\Gamma,$ and a weight function $w$ on the edges  such that the extension complexity of $P(T,\Gamma,w)$ is superpolynomial in~$n$.
\end{theorem}
\begin{proof}
    Let $S=\{a_1,\ldots,a_n\}$ be a set of positive integers, and let $B\in\mathbb{N}$ be such that $\sum_{i\in[n]}s_i=2B$. So $(S,B)$ is an instance of the partition problem. That is, we are interested in subsets $S'\subset [n]$ such that $\sum_{i\in S'}a_i=B.$ We will show that $\mathrm{PARTITION(S,B)}$ is a face of $P(G,\Gamma,w)$ for a suitable graph $G$, a set of terminals $\Gamma,$ and weights on its edges $w:E(G)\to\mathbb{R}.$ In fact, the graph that we build is a tree. Thus, for each instance $(S,B)$ of the partition problem for which $\mathrm{PARTION}(S,B)$ has superpolynomial extension complexity, we get a polytope for the min-max connected multiway cut that has superpolynomial extension complexity.

    We build a weighted tree as follows. We have a special vertex called the root terminal in $T.$ For each element $a_i\in S$ we attach a five-vertex subtree consisting of four terminals and one non-terminal, as shown in Figure \ref{fig:xc}. 
    \begin{figure}[htbp]
        \centering
        \includegraphics[width=0.35\textwidth]{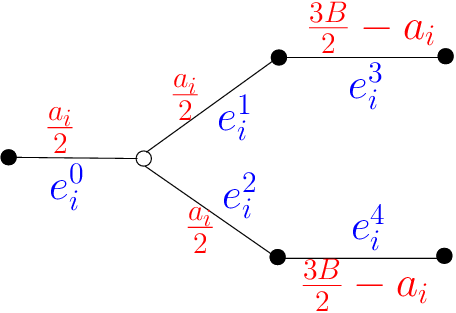}
        \caption{Superpolynomial extension complexity gadget; solid circles are terminals, hollow circles are non-terminals}
        \label{fig:xc}
    \end{figure}
We label the edges $e_i^0, e_i^1, e_i^2, e_i^3,$ and $e_i^4$ as shown. The weights of the edges are as follows: $w(e_i^0)=w(e_i^1)=w(e_i^2)=a_i/2, w(e_i^3)=w(e_i^4)=3B/2-a_i.$

 Each connected cut in $T$ has a value at least $3B/2$, so $z=3B/2$ is a face of $P(T,\Gamma,w)$ where $z$ is the last coordinate. Let us call this face $F$. For any connected cut $C$, consider the subset $S(C)=\{i~|~e_i^0\notin C\}.$ For the cut $C$ we see that $\sum_{i=1}^n (w(e_i^1)x_{e_i^1}+w(e_i^2)x_{e_i^2}) = B+w(S(C))/2$ which also equals the size of the boundary of the root terminal which for any vertex in $F$ is at most $3B/2.$ Therefore, consider the face $G$ of $F$ given by the hyperplane $\sum_{i=1}^n (w(e_i^1)x_{e_i^1}+w(e_i^2)x_{e_i^2})=3B/2.$ $G$ contains exactly those connected cuts in $T$ that correspond to subsets with sum exactly equal to $B.$ The projection map given by $y_i=1-x_i^0$, $\forall i\in [n]$, maps $G$ to the partition polytope of instance $(S,B).$ Therefore $\textrm{PARTITION}(S,B)$ is projection of a face of $P(T,\Gamma,w)$ and so by Theorem \ref{thm:xc_lb_partition} we have, for infinitely many $n$, a tree and weights so that the extension complexity of the corresponding min-max connected multiway cut polytope is superpolynomial in $n$.
\end{proof}

\subsection{Weighted Exact Min-Max Connected Cut}
Finally, we show that finding a connected multiway cut of specified size is NP-hard for trees. 
\begin{theorem}
    Given a weighted tree and $B\in\mathbb{N}$ it is NP-hard to decide whether there is a connected cut of cost exactly $B$.
\end{theorem}
\begin{proof}
    Given an instance of the partition problem $(S,B)$, we consider the tree used in the proof of Theorem~\ref{thm:xc_lb}. By adding an additional terminal attached to the root terminal by an edge with weight $2B$, we can ensure that the terminal with the maximum boundary size is the root terminal. Furthermore, for any connected cut $C$ and the corresponding subset $S(C)$, the size of the boundary of the root terminal is exactly $3B+w(S(C))/2.$ Thus, one can check whether the partition instance is a yes instance or not by checking whether there is a connected sum with value $7B/2.$
\end{proof}

\section{Algorithms}\label{sec:algorithms}
Given the $W[1]$-hardness of the problem on graphs of bounded treewidth, under the standard 
assumption that FPT$\not=W[1]$, an XP algorithm parameterized by the treewidth is the best 
one can hope to achieve for these graphs. We describe such an algorithm, and a few other algorithmic results, 
in this section.

\subsection{XP Algorithms for Bounded Treewidth Graphs}\label{subsec:xpAlg}
We begin by giving an informal overview of the algorithm. The algorithm performs dynamic programming on a nice tree decomposition, storing at each node $\alpha$ a set of feasible configurations that summarize partial solutions (i.e., a partition) on the corresponding subgraph. 
Each configuration encodes a partition of the bag $B(\alpha)$, boundary sizes and terminal information of parts intersecting with $B(\alpha)$, an upper bound on boundary sizes of the other parts belonging to the current subgraph, and adjacency relations among parts
intersecting $B(\alpha)$.
Starting from leaf nodes, configurations are propagated bottom-up using rules for introduce, forget, and join nodes that combine or update partitions. It is guaranteed that among the constructed configurations at each node, the one corresponding to an optimal global solution is preserved. At the root, selecting the configuration minimizing the maximum boundary value yields an optimal solution, giving an XP-time algorithm.

The crucial observation here is that at any tree node $\alpha$ for any partition of vertices, if a part uses vertices (in bags of the nodes) completely below $\alpha$ then these parts are essentially frozen as far as nodes above $\alpha$ are concerned, so we do not need to keep track of them. And, for parts that may change in future as more vertices are considered (as we move up the tree), these must be connected ``via the nodes in the current bag''. But then we need to consider only the information relevant to the current bag.

Now we describe the algorithm more precisely.

\medskip

A \emph{tree decomposition} of a graph $G=(V,E)$~\cite{Kloks:94} 
is a pair $(T, B)$, where $T$ is a rooted tree and $B$ is a mapping
$B: V(T) \rightarrow 2^V$ satisfying
\begin{itemize}
        \item for any $uv \in E$, there exists $\alpha \in V(T)$ such that
        $u, v \in B(\alpha)$,
        \item if $v \in B(\alpha)$ and $v \in B(\beta)$, then $v \in B(\gamma)$ for all
        $\gamma$ on the path from $\alpha$ to $\beta$ in $T$.
\end{itemize}
We use the convention that the vertices of the tree are called nodes and are labeled with small Greek letters, 
and the sets $B(\alpha)$ are called \emph{bags}.

The {\em treewidth $tw(T, B)$ of a tree decomposition} $(T, B)$ is
the size of the largest bag of $(T, B)$ minus one.
The {\em treewidth $tw(G)$ of a graph} $G$ is the
minimum treewidth over all possible tree decompositions of $G$.

A \emph{nice tree decomposition}~\cite{Kloks:94} 
is a tree decomposition with one
special node $\rho$ called the \emph{root}, in which each node is
one of the following types:
\begin{itemize}
        \item \emph{Leaf node}: a leaf $\alpha$ of $T$ with $B(\alpha) = \emptyset$.
        \item \emph{Introduce node}: an internal node $\alpha$ of $T$ with one
        child $b$ for which $B(\alpha) = B(\beta) \cup \{v\}$ for some $v \not\in
        B(\alpha)$.
        \item \emph{Forget node}: an internal node $\alpha$ of $T$ with one child
        $b$ for which $B(\alpha) = B(\beta) \setminus \{v\}$ for some $v \in B(\beta)$.
        \item \emph{Join node}: an internal node $\alpha$ with two children $b$
        and $c$ with $B(\alpha) = B(\beta) = B(\gamma)$.
\end{itemize}

For any graph $G$ of treewith $\tau$ on $n$ vertices, a nice tree decomposition of $G$ of width $\tau$
with at most $8n$ nodes can be computed
in time $f(\tau) \cdot n$, for some computable function $f$~\cite{Bodlaender:93,Kloks:94}.

Given a rooted tree decomposition $(T, B)$ and a node $\alpha \in V(T)$,
we denote by $T_\alpha$ the subtree of $T$ rooted in $\alpha$, 
by $V_\alpha=\bigcup_{\beta \in V(T_\alpha)} B(\beta)$ the union of bags of all nodes
in $T_\alpha$, and by $G_\alpha=G[V_\alpha]$ the subgraph of $G$ induced by $V_\alpha$. 

For a set $V$, its subset $S$, and a partition $\CP\in \Pi(V)$,
the \emph{restriction} of $\CP$ on $S$ is the partition 
$\CP|_S=\{X\cap S \mid X\in \CP \mbox{ and } X\cap S\not = \emptyset\}$ of $S$. 
A partition $\CP'\in \Pi(V)$ is an \emph{extension of}
a partition $\CP\in \Pi(S)$ if $\CP=\CP'|_S$; in this case, if $U\in \CP$,
$U'\in \CP'$ and $U\subseteq U'$, we also say that the set $U'$ is an
\emph{extension of the set} $U$. 

Let $G=(V,E)$ be a graph, $\Gamma=\{1,\ldots,k\}\subseteq V$ a set of terminals, 
and $(T,B)$ a nice tree decomposition of $G$ of treewidth $\tau$, rooted at a node $\rho$.
For a node $\alpha\in V(T)$, let $\Gamma_\alpha=\Gamma\cap V_\alpha$, and let $\bot$ be 
a symbol not appearing in $V$. 
A \emph{configuration of a node} $\alpha\in V(T)$ with a bag $B(\alpha)$, is a 
tuple $(D,\CP,s,b,a)$,  
where $D\in [m]$, $\CP\in \Pi(B(\alpha))$, $s\in (\Gamma_\alpha\cup\{\bot\})^\CP$, $b\in [m]^\CP$, 
$a\in \{0,1\}^{\binom{\CP}{2}}$,
and for each $U\in \CP$, it holds that $U\cap \Gamma_\alpha\subseteq \{s_U\}$.
The configuration $(D,\CP,s,b,a)$ is \emph{feasible} if there exists an extension  
$\CP'\in \Pi(V_\alpha)$ of the partition $\CP$ 
such that 
\begin{enumerate}
    \item for each $U'\in \CP'$, $G[U']$ is connected and $|U'\cap \Gamma_\alpha|\leq 1$,  
    \item for each $U'\in \CP'$, if $U'\cap B(\alpha)=\emptyset$, then $E(U',V_\alpha\setminus U') \leq D$ and $|U'\cap \Gamma_\alpha|=1$,
    \item for each $U'\in \CP'$, if $U'\cap B(\alpha)\not =\emptyset$, then $E(U',V_\alpha\setminus U')=b_U$,
    \item for each $(U,W)\in \binom{\CP}{2}$, $a_{(U,W)}=0$ if $E(U,W)=\emptyset$, and $a_{(U,W)}=1$ otherwise.
\end{enumerate}

The intended meaning of the feasible configuration $(D,\CP,s,b,a)$ of the node $\alpha$,
and of the extension of $\CP'$ of $\CP$ is as follows: 
The number $D$ is an upper bound on the largest boundary size of those parts of the partition $\CP'$ that do not intersect with $B(\alpha)$. The partition $\CP$ is a restriction of the partition $\CP'$ of $V_\alpha$ to $B(\alpha)$. 
For each $U\in \CP$, $s_U$ specifies the unique terminal in the extension $U'\in \CP'$
of $U$, if there is one; otherwise its value is $\bot$. 
For each $U\in \CP$, the number $b_U$ gives the boundary size of the extension 
$U'\in \CP'$ of $U$, and for each $(U,W)\in \binom{\CP}{2}$, $a_{(U,W)}$ captures
information about the adjacency of the extensions $U'$ and $W'$ of $U$ and $W$:
the value is $1$ if there is an edge $uv\in E(G_\alpha)$ with $u\in U$ and $v \in W$,
and is $0$ otherwise.

The following lemma follows immediately from the definition of a feasible configuration.
\begin{lemma}\label{lem:solution}
If $(D,\CP,s,b,a)$ is a feasible configuration of the root $\rho$ of $T$ such that
for each $U\in \CP$, $s_U\not=\bot$, then there
exists a solution of the given min-max connected multiway cut instance of cost 
$\max\{D,\max_{U\in \CP} b_U\}$.
\end{lemma}

From now on, let  $\CP^\star=\{V_1,\ldots,V_k\}$ denote a fixed optimal partition for the instance
$G=(V,E)$, $\Gamma$ of the min-max connected multiway cut. Given a subset $U\subseteq V$ 
satisfying $U\subseteq V_i$, for some part $V_i$ of $\CP^\star$, we denote by $\bar U$ the 
corresponding part $V_i$ (which is the extension of $U$).
For each node $\alpha$ of the tree decomposition $(T,B)$ of $G$, we define the following quantities: 
\begin{itemize}
    \item $D_\alpha=\max \{|E(U,V_\alpha\setminus U)| \mid U\in \CP^\star \mbox{ s.t. } U\cap B(\alpha)=\emptyset \mbox{ and } U\subseteq V_\alpha\}$\footnote{By convention, for purposes of this section, we define $\max \emptyset=0$.},
    \item $\CP_\alpha=\CP^\star|_{B(\alpha)}$,
    \item for each $U\in \CP_\alpha$, if $\Gamma\cap V_\alpha\cap \bar U$ is non-empty, let $s_U$ be the unique terminal in it, and let $s_U=\bot$ otherwise,
    \item for each $U\in \CP_\alpha$, let $b_U=|E(\bar U\cap V_\alpha, V_\alpha\setminus \bar U)|$,
    \item for each $(U,W)\in \binom{\CP_\alpha}{2}$, let $a_{(U,W)}=1$ if $E(\bar U\cap V_\alpha,V_W\cap V_\alpha)>0$,
    and let $a_{(U,W)}=0$ otherwise.
\end{itemize}
A routine verification shows that $C_\alpha=(D_\alpha, \CP_\alpha, s, b, a)$ is a feasible configuration 
of the node $\alpha$ and we call it an \emph{optimal configuration} of $\alpha$.

Consider a join node $\alpha$ of $T$ with children $\beta$ and $\gamma$, and let
$C_\beta=(D_\beta, \CP', s', b', a')$ and 
$C_\gamma=(D_\gamma, \CP'', s'', b'', a'')$ 
be configurations of the nodes $\beta$ and $\gamma$.
The two configurations $C_\beta$ and $C_\gamma$ are \emph{compatible}, if 
$\CP'=\CP''$, and for each $U\in \CP'$, $s'_U=\bot$, $s''_U=\bot$, or $s'_U=s''_U$.

\begin{figure}[hbt]
  \centering
  \includegraphics[width=0.9\textwidth]{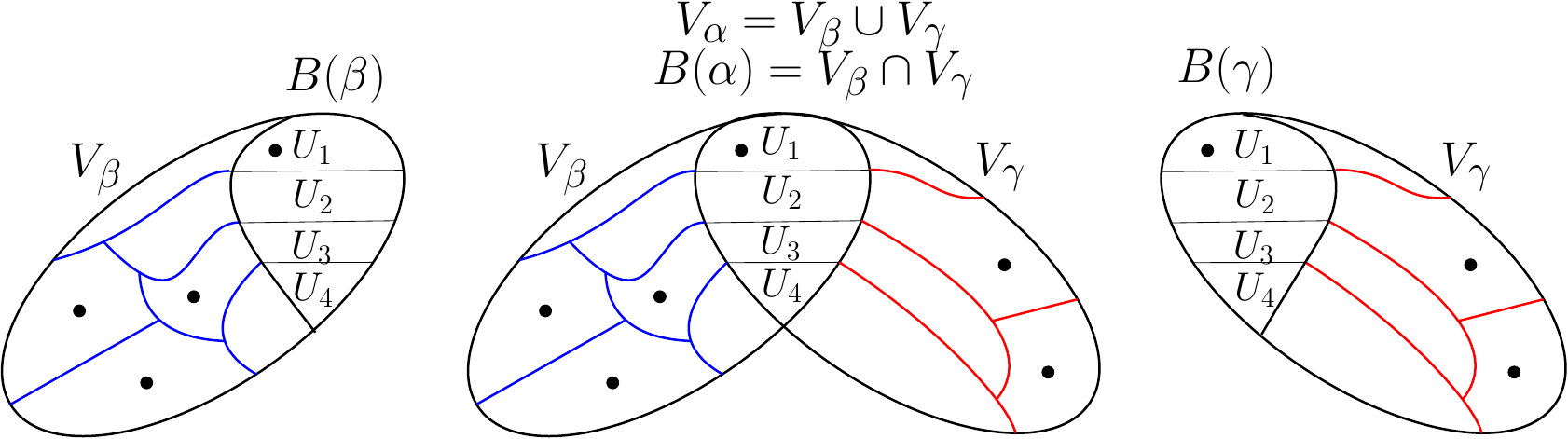}
  \caption{A join node $\alpha$ with children $\beta$ and $\gamma$.
  The left figure depicts the partition $\CP_\beta$ of $V_\beta$, the right figure the partition $\CP_\gamma$
  of $V_\gamma$,
  and the middle figure the partition $\CP_\alpha$ of $V_\alpha$.
  All three figures also depict the partition $\CP=\{U_1, U_2, U_3, U_4\}$ of $B(\alpha)=B(\beta)=B(\gamma)$.
  The terminals are represented by the black dots.}
  \label{fig:join-node-detailed}
\end{figure}

Let $\alpha$ be a join node of $T$ with children $\beta$ and $\gamma$, and
$C_\beta=(D_\beta, \CP, s', b', a')$ and 
$C_\gamma=(D_\gamma, \CP, s'', b'', a'')$ 
their compatible feasible configurations.
The \emph{join derivation} from the two compatible configurations $C_\beta$ and $C_\gamma$ is the configuration
$C_\alpha=(D, \CP, s, b, a)$ defined as follows: 
\begin{itemize}
    \item $D=\max\{D_\beta, D_\gamma\}$, and 
    \item for each $U\in \CP$, $s_U=s'_U$ if $s'_U\in \Gamma_\alpha$, $s_U=s''_U$ if $s''_U\in \Gamma_\alpha$, and $s_U=\bot$ otherwise, and
    \item for each $U\in \CP$, $b_U=b'_U+b''_U-|E(U,B(\alpha)\setminus U)|$, and 
    \item for each $(U,W)\in\binom{\CP}{2}$, $a_{(U,W)}=1-(1-a'_{(U,W)})\cdot (1- a''_{(U,W)})$.
\end{itemize}

\begin{lemma}[Join configuration derivation]\label{lem:join}
The join derivation $C_\alpha=(D, \CP, s, b, a)$ from feasible compatible configurations $C_\beta$ and $C_\gamma$ 
is a feasible configuration of the node $\alpha$.
\end{lemma}

\begin{proof}
As $C_\beta$ and $C_\gamma$ are compatible, the tuple $s$
satisfies the properties required on $s$ in the definition of a configuration; 
as all the other properties are trivially satisfied,
we conclude that the tuple $C_\alpha$ is a configuration.
Now we are going to verify the feasibility of $C_\alpha$.

Let $\CP_\beta\in \Pi(V_\beta)$ and $\CP_\gamma\in \Pi(V_\gamma)$ be extensions of the partition $\CP$ of $B(\alpha)$
to a partition of $V_\beta$ and $V_\gamma$, resp., certifying that $C_\beta$ and $C_\gamma$ are feasible. 
For each $U\in \CP$, by $U_\beta$ ($U_\gamma$, resp.) we denote the unique part 
$U_\beta$ from the partition $\CP_\beta$ ($U_\gamma$ from $\CP_\gamma$, resp.)
satisfying $U_\beta\cap{B(\alpha)}=U$ ($U_\gamma\cap {B(\alpha)}=U$, resp.).
Let $\CP_\alpha$ be the partition of $V_\alpha$ defined by 
\begin{align}
\CP_\alpha= \{U\in \CP_\beta \mid U\cap B(\alpha)=\emptyset \} \cup
\{U\in \CP_\gamma \mid U\cap B(\alpha)=\emptyset \} \cup 
\{U_\beta \cup U_\gamma \mid U\in \CP \} \ .
\end{align}
By construction, $\CP_\alpha$ is an extension of $\CP$. Moreover,
\begin{enumerate}
    \item by construction, for each $U\in \CP_\alpha$, $G[U]$ is connected, and
    \item for each $U\in \CP_\alpha$, if $U\cap B(\alpha)=\emptyset$, then 
    $U\in \CP_\beta\cup \CP_\gamma$, thus, by feasibility of $C_\beta$ and $C_\gamma$,  $E(U,V_\alpha\setminus U)\leq D$, and by their compatibility, $|U\cap \Gamma_\alpha|=1$
    as $\Gamma_\alpha=\Gamma_\beta\cup \Gamma_\gamma$, and
    \item the third and the fourth properties follow from the definition of the tuples $b$ and $a$. 
\end{enumerate}
\vspace{-.45cm}
\end{proof}

Let $\alpha$ be a forget node of $T$ with a child $\beta$ for which $B(\alpha)=B(\beta)\setminus \{v\}$, and let 
$C_\beta=(D, \CP, s, b, a)$ be a feasible configuration of the node $\beta$.
The \emph{forget derivation} from the configuration $C_\beta$ is the configuration 
$C_\alpha=(D', \CP', s', b', a')$ 
defined as follows: 

\begin{itemize}
\item If $\{v\}\in \CP$, then 
$D'=\max\{D,b_{\{v\}} \}$,
$\CP'=\CP\setminus \left\{\{v\}\right\}$,  
and for each $U\in \CP'$, $s'_U=s_U$, $b'_U=b_U$,
and for each $(U,W)\in \binom{\CP'}{2}$, $a'_{(U,W)}=a_{(U,W)}$.

\item Otherwise, 
$D'=D$, $\CP'=\CP|_{B(\alpha)}$, 
for each $U\in \CP'\cap \CP$, $s'_U=s_U$, $b'_U=b_U$, 
for each $(U,W)\in \binom{\CP'\cap \CP}{2}$
and for $Z$ denoting the part of $\CP'$ such that $Z\cup \{v\}\in \CP$, 
$s'_{Z}=s_{Z\cup \{v\}}$, $b'_{Z}=b_{Z\cup \{v\}}$, and
for each $U\in \CP'$, $a'_{(Z,U)}=a_{(Z\cup \{v\},U)}$.
\end{itemize}
\begin{lemma}[Forget configuration derivation]\label{lem:forget}
The forget derivation $C_\alpha=(D', \CP', s', b', a')$ from feasible configuration $C_\beta$ 
is a feasible configuration of the node $\alpha$.
\end{lemma}
\begin{proof}
Let $\CP_\beta\in \Pi(V_\beta)$ be an extension of $\CP$ that certifies the feasibility of $C_\beta$. As $V_\alpha=V_\beta$, we also have $\CP_\beta\in \Pi(V_\alpha)$. 
Moreover, by construction of $\CP'$, the partition $\CP_\beta$ is also an extension of 
$\CP'$; it is immediate from the definition of $C_\alpha$ that it satisfies all
the remaining requirements of feasibility as well.
\end{proof}

Let $\alpha$ be an introduce node of $T$ with a child $\beta$ for which $B(\alpha)=B(\beta)\cup \{v\}$, and let 
$C_\beta=(D, \CP, s, b, a)$ be a feasible configuration of the node $\beta$.
Further, let
$\CQ\subseteq \CP$ 
be a subset of $\CP$ satisfying properties
\begin{enumerate}[label=\textbf{Q\arabic*.}, leftmargin=2.5em]
\item for each $U\in \CQ$, $\delta(v)\cap U\not = \emptyset$, and 
\item for each $(U,W)\in \binom{\CQ}{2}$, $a_{(U,W)}=0$,  and 
\item $|\bigcup_{U\in \CQ}\{s_U\}\setminus\{\bot\}|\leq 1$,
\end{enumerate}
\smallskip
and let $\bar U=\bigcup_{U\in \CQ}U\cup \{v\}$.
The \emph{introduce derivation from $C_\beta$ and $\CQ$} is the configuration 
$C_\alpha=(D', \CP', s', b', a')$ defined as follows:
\begin{itemize}
\item $D'=D$,
\item $\CP'=(\CP\setminus \CQ)\cup \{\bar U\}$,
\item for each $U\in \CP\setminus \CQ$, $s'_U=s_U$, $b'_U=b_U+|E(U,\{v\})|$, 
\item $s'_{\bar U}\in\bigcup_{U\in \CQ}\{s_U\}$
if $|\bigcup_{U\in \CQ}\{s_U\}\setminus\{\bot\}|= 1$, and $s'_{\bar U}=\bot$ otherwise,
\item $b'_{\bar U}=\sum_{U\in \CQ} b_U + |E(B(\alpha)\setminus \bar U,\{v\})|$, and
\item 
for each $(U,W)\in \binom{\CP\setminus \CQ}{2}$, $a'_{(U,W)}=a_{(U,W)}$, for each
$U\in \CP\setminus \CQ$, $a'_{(U,\bar U)}=\min\{1,\sum_{W\in \CQ}a_{(U,W)} \}$.
\end{itemize}
\begin{lemma}[Introduce configuration derivation]\label{lem:introduce}
Given a feasible configuration $C_\beta=(D, \CP, s, b, a)$, 
for each $\CQ\subseteq \CP$ satisfying the properties Q1, Q2 and Q3, the introduce derivation 
$C_\alpha=(D', \CP', s', b', a')$ from $C_\beta$ and $\CQ$
is a feasible configuration of the node $\alpha$.
\end{lemma}
\begin{proof}
By construction, $\CP'$ is a partition of $B(\alpha)$.
Let $\CP_\beta\in \Pi(V_\beta)$ be an extension of $\CP$ that certifies the feasibility of $C_\beta$.
Let 
\[
U'=\bigcup_{\substack{U\in\CP_\beta \\U\cap \bar U\not = \emptyset}}U\cup \{v\} \ ,
\]
that is, the set defined by the union of all the parts in $\CP_\beta$ intersecting $\bar U$
together with the vertex $v$,
and let 
$\CP_\alpha$ be the partition of $V_\alpha$ defined by 
\begin{align}
\CP_\alpha= \{U\in \CP_\beta \mid U\cap \bar U=\emptyset \} \cup \{U'\} \ .
\end{align}
Then, $\CP_\alpha$ is an extension of $\CP'$.
A routine verification shows that $\CP_\alpha$ certifies the feasibility of $C_\alpha$.
\end{proof}

For ease of reference, we state the following lemma about configurations of leaves. 
Since the empty bag admits only the empty partition, the claim is immediate.
\begin{lemma}[Leaves’ configurations]\label{lem:leaves}
For each leaf $\alpha$ of $T$, $\CC=(0,\emptyset,\emptyset,\emptyset,\emptyset)$ is a feasible configuration.
Moreover, $\CC$ is an optimal configuration.
\end{lemma}

\begin{lemma}[Derivation of optimal configurations]\label{lem:optimal}
1. If $\alpha$ is a join node with children $\beta$ and $\gamma$, and $C_\alpha$, $C_\beta$
and $C_\gamma$ are their optimal configurations, 
then $C_\alpha$ is the join derivation from $C_\beta$ and $C_\gamma$.

2. If $\alpha$ is a forget node with a child $\beta$, and $C_\alpha$ and $C_\beta$
are their optimal configurations,
then $C_\alpha$ is the forget derivation from $C_\beta$. 

3. If $\alpha$ is an introduce node with a child $\beta$, and $C_\alpha=(D',\CP',s',b',a')$ 
and $C_\beta=(D,\CP,s,b,a)$
are their optimal configurations, then there exists a subset $\CQ\subseteq \CP$ satisfying
the properties Q1 - Q3 such that $C_\alpha$ is the introduce derivation 
from $C_\beta$ and $\CQ$.
\end{lemma}
\begin{proof}
A routine verification shows that the parts 1 and 2 hold. For part 3, let $W\in \CP'$ be the
part of $\CP'$ containing the vertex $v\in V$ introduced in the node $\alpha$, and let 
$\CQ=\{U\in \CP \mid U\subseteq W \}$; then again a routine verification shows that the claim
holds.
\end{proof}

Given Lemmas~\ref{lem:solution}-\ref{lem:optimal}, we are ready to describe our
algorithm and prove its correctness. We start with the leaves of $T$,
and for each of them construct the optimal configuration (Lemma~\ref{lem:leaves}). Then
we procceed in a bottom-up fashion: For each join node $\alpha$ with children $\beta$ and 
$\gamma$, given a set of feasible configurations for $\alpha$ and a a set of feasible 
configurations for $\beta$, for each compatible pair of feasible configurations $C_\beta$ and
$C_\gamma$, we derive a feasible configuration for $\alpha$ by Lemma~\ref{lem:join}.

Similarly, for each forget node $\alpha$ with a child $\beta$, from each feasible configuration
$C_\beta$ of the node $\beta$, we derive a feasible configuration $C_\alpha$
of the node $\alpha$ by Lemma~\ref{lem:forget}.

Further, for each introduce node $\alpha$ with a child $\beta$, for each feasible configuration
$C_\beta=(D,\CP,s,b,a)$ 
of the node $\beta$, and for each subset $\CQ\subseteq \CP$
satisfying properties Q1-Q3, we derive
a feasible configuration of the node $\alpha$ by Lemma~\ref{lem:introduce}.

By Lemmas~\ref{lem:leaves} and~\ref{lem:optimal}, by this process, for each node of the
tree, one of its feasible configurations, constructed by the previously described steps, 
will be the optimal one. Thus, by picking
a feasible configuration $C=(D,\CP,s,b,a)$ of the root for which the value
$\max\{D, \max_{U\in \CP} b_U \}$ is minimal, we get, with the help of Lemma~\ref{lem:solution},
an optimal solution. 

Concerning the running time of the algorithm, we observe that the number
of feasible configurations of a node is at most 
$K=m(\tau+1)^{\tau+1}(k+1)^{\tau+1} m^{\tau+1} 2^{\binom{\tau+1}{2}}$ where $k=|\Gamma|$;
clearly $m$ is an upper bound on cost of an optimal solution.
Thus, the processing time of a join node is $\Oh{K^2}$ at most, the processing time of a forget
node is $\Oh{K}$ at most, and the processing time of an introduce node is $\Oh{K 2^\tau}$ at most.
Hence, the overall running time is polynomial in $m$, with degree linear in $\tau$.
The following theorem summarizes the result.


\begin{theorem}\label{thm:XPalg}
The unweighted min-max connected multiway cut on graphs with treewidth~$\tau$ is solvable
in time $2^{\Oh{\tau^2}}n^{\Oh{\tau}}$.
\end{theorem}

If the weights are non-negative integers, each edge of weight $w$ can be replaced by $w$ parallel paths, each 
of length two.
This transformation increases the graph size by a factor proportional to
$w$ but does not increase treewidth. Thus, we have a pseudopolynomial algorithm for the weighted version of the problem.

\begin{corollary}\label{cor:XPalg}
The weighted min-max connected multiway cut on graphs with treewidth $\tau$ is solvable
in time $2^{\Oh{\tau^2}}N^{\Oh{\tau}}$ where $N=\sum_{e\in E}w(e)$ is the sum of all edge weights.
\end{corollary}

\subsection{FPTAS for Bounded Treewidth Graphs} 
In this subsection, we show that even with unbounded weights, the weighted min-max connected 
multiway cut problem can be approximated efficiently on graphs of bounded treewidth. 

\begin{theorem}\label{thm:fptas}
For every $\varepsilon>0$, there is an algorithm for the weighted min-max connected multiway cut on graphs of treewidth $\tau$ that runs in time polynomial in the instance size, and $1/\varepsilon$ and finds a $(1 + \varepsilon)$-approximation solution. The degree of the polynomial is linear in $\tau.$
\end{theorem}
\begin{proof}
The algorithm uses the standard rounding and scaling approach. 

Given $\varepsilon>0$, a graph $G=(V,E)$ of treewidth $\tau$ with a weight function $w:E\to \NN$ and a set of 
terminals $\Gamma\subseteq V$, we start by checking
whether there exists an optimal min-max connected cut of value at most $m/\varepsilon$.
This can be done in time polynomial in $n$ and $1/\varepsilon$ by the algorithm of Corollary~\ref{cor:XPalg}, 
after a straightforward modification of the input instance;
the degree of the polynomial depends linearly on $\tau$. If there is such a cut, we have found an optimal solution in time polynomial in $n$ and $1/\varepsilon$ and we stop.

From now on, we assume that $OPT>m/\varepsilon$. Let $W=\max_{e\in E} w(e)$
and $d=\frac{\varepsilon W}{m}$.
For each edge $e\in E$, we define a new edge weight $w'(e)=\lceil \frac{w(e)}{d}\rceil$.
Note that the new edge weights are integers in the range from $0$ to $\lceil \frac{m}{\varepsilon}\rceil$. 
Now the algorithm from Corollary~\ref{cor:XPalg}
can be used to find, in time polynomial in the instance size $n$ and $1/\varepsilon$, an optimal solution 
for the modified instance where the degree of the polynomial depends linearly on $\tau$; this solution will be presented as an output for the original weighted instance
of the problem.

Consider now the approximation ratio. Let $F\subseteq E$ be an optimal min-max connected cut
and $F_{\max}\subseteq F$ be the largest (w.r.t. to the weights $w$) boundary of a component 
in the optimal solution. Similarly, let $A\subseteq E$ be the connected cut
constructed by the algorithm of Corollary~\ref{cor:XPalg}, and $A_{\max}\subseteq A$ the largest boundary of
a component in this solution. Then, as for each edge $e\in E$, $w(e)\leq d\cdot w'(e)\leq w(e)+1$,
we have
\[
\sum_{e\in A_{\max}}w(e) \leq d\cdot \sum_{e\in A_{\max}}w'(e) \leq d\cdot \sum_{e\in F_{\max}}w'(e)
\leq \sum_{e\in F_{\max}}w(e) + |F_{\max}| \leq OPT + m \leq (1+\varepsilon)\cdot OPT \ .
\]
\end{proof}
\subsection{FPT Algorithm for trees}
In Section \ref{sec:hardness} we proved that the min-max connected multiway cut is weakly NP-hard on graphs of treewidth at least two, and $W[1]$-hard when parameterized by the number of terminals on general graphs. We complement these results now by showing that on trees the problem admits an FPT algorithm when parameterized by the number of terminals.

\begin{theorem}
The weighted min-max connected multiway cut on trees is fixed-parameter tractable with respect to the number of terminals.
\end{theorem}
\begin{proof}
The essence of the proof is the kernelization method. Given an instance of the problem, 
that is, a tree $T=(V,E)$ and a set of terminals $\Gamma\subseteq V$, we first choose any of the terminals as the root of the tree. Without loss of generality, we can assume that every leaf of the tree $T$
is a terminal; otherwise, we remove from $T$ all vertices (and the adjacent edges) that do not 
have any terminal as their descendant.

Assuming that all leaves of $T=(V,E)$ are terminals, let $W$ be the subset consisting of all
terminals and all vertices of degree at least three. We say that two vertices $u,v\in W$ are
\emph{neighboring} if the unique path between $u$ and $v$ in $T$ does not contain any vertex
from $W$ as an internal vertex. For any two neighboring vertices $u$ and $v$, we replace the
path between $u$ and $v$ by a single edge, and the weight of the edge is set to the minimum 
weight of the edges on the original $u-v$ path in $T$. Let $T'=(V',E')$ denote the resulting
tree. 

We observe two things:
\begin{itemize}
\item A connected multiway cut for $T$ can be easily transformed into a connected
multiway cut for $T'$ of the same value, and vice versa.
\item The tree $T'$ has at most $2|\Gamma|-1$ vertices and at most $2|\Gamma|-2$ edges. 
\end{itemize}

Thus, we can solve the weighted min-max connected multiway cut problem in $T'$ optimally in time 
$\Oh{2^{2|\Gamma|}}$, and from the optimal solution obtain an optimal solution for $T$.    
\end{proof}

\section{Concluding Remarks}\label{sec:concl}




Motivated by its close connection to the well-studied Spanning Tree Congestion problem, 
we have proposed a natural variant of the multiway cut problem 
where each part in the partition is required to be connected, and 
the objective is to minimize the maximum cost of a part rather than the total cost.
We shown various hardness results for the 
problem and complement them by giving algorithms in various cases.

We see that the problem is weakly NP-hard on graphs with treewidth at least two, and when parameterized by the treewidth the problem admits a ``weakly XP'' algorithm, that is, an algorithm with running time $2^{\Oh{\tau^2}}N^{\Oh{\tau}}$ where $\tau$ is the treewidth and $N$ is the input size when weights are encoded in unary. Some dependence on $\tau$ in the exponent of $N$ is perhaps unavoidable because the problem is shown to be $W[1]$-hard when parameterized by the treewidth.

One gap in our results concerns trees. We only have a pseudopolynomial algorithm for the weighted problem on trees but we do not know whether on trees the min-max connected multiway problem is weakly NP-hard. 

Another important open problem concerns general graphs: to design an approximation algorithm with a non-trivial approximation ratio and to establish corresponding hardness results.


\bibliography{connected}
\end{document}